\documentclass[12pt,a4paper]{article}
\usepackage[left=2.5cm,right=2.5cm]{geometry}

\usepackage{amsmath,amsthm,amsfonts}
\usepackage{hyperref}

\newtheorem{thm}{Theorem}

\newtheorem{prop}{Proposition}
\newtheorem{lemma}{Lemma}
\newtheorem{defn}{Definition}
\newtheorem{rem}{Remark}

\newcommand{\R}{\mathbb{R}}
\newcommand{\C}{\mathbb{C}}
\newcommand{\s}{\mathbb{S}}
\newcommand{\rme}{\mathrm{e}}
\newcommand{\rmi}{\mathrm{i}}
\newcommand{\rmd}{\mathrm{d}}

\newcommand{\ohbig}{\mathcal{O}}

\newcommand{\norm}[1]{\left\|#1\right\|}

\newcommand{\F}{\mathcal{F}}
\newcommand{\comp}{\text{comp}}

\begin{document}

\title{Inverse fixed energy scattering problem for the two-dimensional nonlinear Schr\"odinger operator}
\author{Georgios Fotopoulos and Valery Serov\\
Department of Mathematical Sciences, University of Oulu, \\ PO Box 3000, FIN-90014 Oulu, Finland\\georgios.fotopoulos@oulu.fi and vserov@cc.oulu.fi}
\date{}
\maketitle

\begin{abstract}
This work studies the direct and inverse fixed energy scattering problem for  the two-dimensional  Schr\"odinger equation with rather general nonlinear index of refraction. In particular, using the Born approximation we prove that all singularities of the unknown compactly supported potential from $L^2$ space can be obtained uniquely by the scattering data with fixed positive energy. The proof is based on the new estimates for the Faddeev Green's function in $L^{\infty} (\R^2)$ space. 
\end{abstract}


\section{Introduction}
We consider the nonlinear generalized Schr\"odinger equation in $\mathbb{R}^{2}$

\begin{equation}\label{schrodinger}
-\Delta u(x) + h(x,|u(x)|)u(x) = k^{2}u(x), \quad k \in \R
\end{equation}
where $k^2>0$ is fixed and the potential function $h$ has some specific properties which we will mention later.

In order to formulate the scattering data we consider scattering solutions of \eqref{schrodinger} of the form
\[u(x,k,\theta)=u_{0}(x,k,\theta)+u_{sc}(x,k,\theta)
\]
where $u_{0}(x,k,\theta)=\rme^{\rmi k(x,\theta)}$ is the incident wave with direction $\theta \in \mathbb{S}^{1}:=\{x\in \R^2:|x|=1\}$ and 
$u_{sc}(x,k,\theta) $ is the scattered wave. The scattered wave must satisfy the Sommerfeld radiation condition at infinity.
These solutions are the unique solutions of the Lippmann--Schwinger equation
\begin{equation}\label{lippman-schwinger}
u(x,k,\theta) = \rme^{\rmi k(x,\theta)} - \int_{\mathbb{R}^{2}}{G_{k}^{+}(|x-y|)h(y,|u|)u(y)\rmd y},
\end{equation}
where $G_{k}^{+}$ is the outgoing fundamental solution of the two dimension Helmholtz equation and is defined as
\begin{equation*}
G_{k}^{+}(|x|) = \frac{\rmi}{4}H_{0}^{(1)}(|k||x|),
\end{equation*}
where $H_{0}^{(1)}$ is the Hankel function of the first kind and zero order. Recall that the function  $G_{k}^{+}$ is the kernel
of the integral operator $(-\Delta - k^{2} - \rmi 0)^{-1}$. 

In \cite{SeHF} it was established the asymptotic form of the scattering solutions for fixed $k\ge k_0>0$ and for all $\theta',\theta \in \s^1$, which gives access to scattering data
$$
u = u_0 - \frac{1+\rmi}{4\sqrt{\pi}}k^{-\frac{1}{2}}|x|^{-\frac{1}{2}}\rme^{\rmi k|x|}A(k,\theta',\theta) + o\Big(\frac{1}{|x|^\frac{1}{2}}\Big), \quad |x|\to \infty,
$$
where $\theta,\theta'=\frac{x}{|x|} \in \s^1$ and the function $A(k,\theta',\theta)$ is called the scattering amplitude and is defined as
\begin{equation*}
A(k,\theta',\theta) := \int_{\mathbb{R}^{2}}{\rme^{-\rmi k(\theta',y)}h(y,|u|)u(y)\rmd y},
\end{equation*}
where $\theta'$ is called the direction of observation.

In two dimensions in the case of linear Schr\"odinger operator the first uniqueness and reconstruction algorithm was proved by Nachman
\cite{nachman} via $\bar{\partial}$-methods for potentials of conductivity type. Note that   Grinevich and Novikov \cite{grinevich-novikov} showed that in two dimensions there are nonzero real potentials of the Schwartz class with zero amplitude at fixed energy. Sun and Uhlmann \cite{sun-uhlmann} proved uniqueness for potentials satisfying nearness conditions to each other. A global uniqueness result for the linear Schr\"odinger equation with fixed energy was proved only in 2008 by Bukhgeim \cite{bukhgeim} for compactly supported potentials from $L^p,p>2$. Bukhgeim's result has since been improved and extended to treat related inverse problems (see for example \cite{blasten,imanuvilov-yamamoto}). For the nonlinear two-dimensional Schr\"odinger operator with very general nonlinear potential function $h$ three very well-known inverse scattering problems were successfully solved by Fotopoulos, Harju and Serov (see \cite{SeHF} and \cite{FHSe}). They considered general scattering, backscattering and fixed angle scattering problems. It turned out that (as we can see in this article) inverse fixed energy scattering problem is much more difficult than the others. The result of the present work slightly generalizes the linear case to special type of the nonlinearity. 

It can be mentioned here that in \cite{serov2012} a nonlinear model (for fixed energy scattering problem) was studied in a three-dimensional case.

 This work is organized as follows. In section \ref{direct} we examine the direct scattering problem and prove the existence and uniqueness of special type of solution, called complex geometrical optics (CGO solution ). Section \ref{inverse} discusses the inverse scattering problem where we apply the Born approximation method to recover the main singularities of the unknown function.
 
We use the following definition for the two-dimensional Fourier transform and its inverse
\begin{align*}
(\F f)(\xi)=&\int_{\R^2}{\rme^{\rmi(x,\xi)}f(x)\rmd x}\\
(\F^{-1}f)(x)=&\frac{1}{4\pi^2}\int_{\R^2}{\rme^{-\rmi(x,\xi)}f(\xi)\rmd\xi}.
\end{align*}
The inner product in the Euclidean space $\R^2$ is $(x,\xi)=x_1\xi_1+x_2\xi_2$ where $x=(x_1,x_2)$ and $\xi=(\xi_1,\xi_2)\in \R^2$.
\section{The direct problem}\label{direct} 

In fixed energy scattering problems crucial role plays the Faddeev Green's function
$$
g_z(x):=\frac{1}{4\pi^2}\int_{\R^2}{\frac{\rme^{-\rmi(x,\xi)}}{\xi^2+2(z,\xi)}\rmd \xi},
$$
where $z\in \C^2$ with $(z,z)=0$. This function $g_z(x)$ is the fundamental solution of the following operator with constant coefficients $$
-\Delta-2\rmi(z,\nabla).
$$
The following basic result which is proved in \cite{serov2012} will be used in what follows.
We assume that $\Omega\subset \R^2$ is a bounded domain and we extend $f$ by zero outside of $\Omega$. 
\begin{prop}\label{conv_est}
For any $\gamma<1$ there exists constant $c_{\gamma}>0$ such that for any $f\in L^2(\Omega)$ and for $|z|>1$,
$$
\norm{g_z\ast f}_{L^\infty(\R^2)} \leq \frac{c_{\gamma}}{|z|^\gamma}\norm{f}_{L^2(\Omega)},
$$
where the symbol $\ast$  denotes the convolution.
\end{prop}
The estimate above allow us to prove the existence of complex geometrical optics (CGO) solutions or sometimes also called exponentially growing solutions, for the homogeneous nonlinear Schr\"odinger equation
\begin{equation}\label{hom_schrodinger}
-\Delta u(x) + h(x,|u(x)|)u(x) = 0.
\end{equation}
Such type of solutions were first introduced by Faddeev in \cite{faddeev}. 
In order to define the CGO solutions we introduce a complex parameter $z$ and look for solutions of the form
\begin{equation}\label{cgo}
u(x,z)=\rme^{\rmi(x,z)}\big(1+R(x,z)\big),
\end{equation}
with $(z,z)=0$ 
that satisfy the equation \eqref{hom_schrodinger}. In \cite{nachman} Nachman proved that in the linear case there exists a unique CGO solution for any $z\in \C^2\setminus \{0\}$ and satisfies some asymptotic condition.

Let us discuss the construction of these solutions. Applying \eqref{cgo} to \eqref{hom_schrodinger} we obtain
\begin{equation}\label{Requation}
-\Delta R - 2\rmi(z,\nabla)R + h(x,|\rme^{\rmi(x,z)}(1+R)|)(1+R) =0
\end{equation}

and using the Faddeev Green's function we can write 
\begin{align}\label{R-lippman-schwinger}
R(x,z)=
-\int_{\R^2}{g_z(x-y)h(y,|\rme^{\rmi(y,z)}(1+R(y,z)|))(1+R(y,z))\rmd y}.
\end{align}

We are now in position to prove the existence and the uniqueness of CGO solution for the equation \eqref{hom_schrodinger}.
We assume that the potential function $h$ is compactly supported in $\Omega\subset\R^2$ and has the following properties:
\begin{enumerate}
\item $|h(x,s)| \leq \alpha(x),$   $\alpha \in L^2(\Omega),$ $s\in \R_+$.

\item $\big|h(x,|\rme^{\rmi(x,z)}(1+R_1)|) - h(x,|\rme^{\rmi(x,z)}(1+R_2)|)\big| \leq \beta(x)\big|R_{1} - R_{2}\big|,\ \beta\in L^2(\Omega)$, for any $R_1,R_2 \in L^\infty(\R^2)$ and for any $z\in \C^2$.
\end{enumerate}
We assume here also that $\|\alpha\|_{L^2(\Omega)}>0$ and $\|\beta\|_{L^2(\Omega)}>0$.
\begin{thm}\label{thm1}
Under the above conditions for the potential function $h$ there exists a constant $C_0>0$ such that for all $|z|\ge C_0$ the equation \eqref{R-lippman-schwinger} has a unique solution in the space $L^{\infty}(\R^2)$ and this solution can be obtained as $\lim_{j\to\infty} R_j$ in $L^{\infty}(\R^2)$ with $R_0=0$ and with 
$$
R_j(x,z):=-\int_{\R^2}{g_z(x-y)h(y,|\rme^{\rmi(y,z)}(1+R_{j-1}(y,z)|))(1+R_{j-1}(y,z))\rmd y},\quad j=1,2,...
$$ 
Moreover the following estimates hold
\begin{equation}\label{R_estimates}
\norm{R}_{L^{\infty}(\R^2)} \le \frac{c}{|z|^\gamma},\quad 
\norm{R-R_j}_{L^{\infty}(\R^2)} \le \frac{c}{|z|^{\gamma(j+1)}},\quad j=0,1,2,...
\end{equation}
with some constant $c>0$ and with $\gamma$ as in Proposition \ref{conv_est}.
\end{thm}

\begin{proof}
We consider the ball $B_{\rho}(0):=\{R \in L^{\infty}(\R^2): \|R\|_{L^{\infty}(\R^2)}\leq \rho \}$, for some $\rho>0$ (which will be defined later) and the operator $T: B_{\rho}(0) \rightarrow B_{\rho}(0)$ by 
\begin{equation*}
T(R):= -\int_{\R^2}{g_z(x-y)h(y,|\rme^{\rmi(y,z)}(1+R(y,z)|))(1+R(y,z))\rmd y}.
\end{equation*}
Our aim is to apply Banach's fixed point theorem for the operator $T$. Therefore we have to show that $T:B_{\rho}(0)\to B_{\rho}(0)$ and also that $T$ is a contraction.
Firstly we will show that $T$ maps $B_{\rho}(0)$ into itself. Let $ R\in B_{\rho}(0)$, applying Proposition \ref{conv_est} we obtain
\begin{align*}
\|T(R)\|_{L^{\infty}(\R^2)}   
&  \leq 
\frac{c_{\gamma}}{|z|^\gamma}\norm{\alpha}_{L^2(\Omega)}(1 +\|R\|_{L^{\infty}(\R^2)})\\
 & \leq 
(1+\rho)\frac{c_{\gamma}}{|z|^\gamma}\norm{\alpha}_{L^2(\Omega)}.
\end{align*}
Since we want $\|T(R)\|_{L^{\infty}(\R^2)} \leq \rho$, we require that 
\begin{equation}
\frac{(1+\rho)c_{\gamma}\norm{\alpha}_{L^2(\Omega)}}{|z|^\gamma} \leq \rho\quad \Longleftrightarrow \quad \rho\ge \frac{c_{\gamma}\norm{\alpha}_{L^2(\Omega)}}{|z|^{\gamma}-c_{\gamma}\norm{\alpha}_{L^2(\Omega)}}
\end{equation}
for $|z|^{\gamma}>c_{\gamma}\norm{\alpha}_{L^2(\Omega)}$.

It remains to show that $T$ is a contraction, i.e. 
$$
\|T(R_{1}) - T(R_{2})\|_{L^\infty(\R^2)}\leq q\|R_{1}-R_{2}\|_{L^{\infty}(\R^2)}, \quad q<1.
$$
Let $R_1,R_2 \in B_\rho(0)$ then we obtain
\begin{align*}
\|T(R_{1}) - T(R_{2})\|_{L^\infty(\R^2)} 
 \leq & \frac{c_{\gamma}}{|z|^\gamma}\norm{\beta}_{L^2(\Omega)}\|R_{1}-R_{2}\|_{L^{\infty}(\R^2)}\\ \quad &+ 
  \rho\frac{c_{\gamma}}{|z|^\gamma}\norm{\beta}_{L^2(\Omega)}\|R_{1}-R_{2}\|_{L^{\infty}(\R^2)}+ \frac{c_{\gamma}}{|z|^\gamma}\norm{\alpha}_{L^2(\Omega)}\|R_{1}-R_{2}\|_{L^{\infty}(\R^2)}\\ 
 =& \frac{c_{\gamma}}{|z|^\gamma}\big((\rho+1)\norm{\beta}_{L^2(\Omega)}+\norm{\alpha}_{L^2(\Omega)}\big)\|R_{1}-R_{2}\|_{L^{\infty}(\R^2)}.
\end{align*}
In order for $T$ to be a contraction we require that  
$$
\frac{c_{\gamma}}{|z|^\gamma}\big((\rho+1)\norm{\beta}_{L^2(\Omega)}+\norm{\alpha}_{L^2(\Omega)}\big) < 1
$$ 
or
\begin{equation}
\rho< \frac{|z|^{\gamma}-c_{\gamma}(\norm{\beta}_{L^2(\Omega)}+\norm{\alpha}_{L^2(\Omega)})}{c_{\gamma} \norm{\beta}_{L^2(\Omega)}} 
\end{equation}
for $|z|^{\gamma}>c_{\gamma}(\norm{\beta}_{L^2(\Omega)}+\norm{\alpha}_{L^2(\Omega)})$. Let us choose now $|z|$ and $\rho$ as follows
\begin{equation}\label{constant for mod_z}
|z|\ge \left(c_{\gamma}(\norm{\beta}_{L^2(\Omega)}+2\norm{\alpha}_{L^2(\Omega)})\right)^{\frac{1}{\gamma}}:=C_0,
\end{equation}
\begin{equation}\label{rho_ineq}
\frac{\norm{\alpha}_{L^2(\Omega)}}{\norm{\beta}_{L^2(\Omega)}+\norm{\alpha}_{L^2(\Omega)}}<\rho<\frac{\norm{\alpha}_{L^2(\Omega)}}{\norm{\beta}_{L^2(\Omega)}}.
\end{equation}
It is not so hard to check that with this choice of $\rho$ and $|z|$ the conditions of (8) and (9) are satisfied.
Hence, due to \textit{Banach's fixed point theorem} (see, for example, \cite{Z}), in any ball $B_{\rho}(0)\subset L^{\infty}(\R^2)$ with $\rho$ from \eqref{rho_ineq} and for all $|z|$ which satisfies \eqref{constant for mod_z}, equation \eqref{hom_schrodinger} has a unique CGO solution. Lastly, it's easy to check that the estimates \eqref{R_estimates} hold.
\end{proof}

\section{The inverse problem }\label{inverse}
Next we define what is called as the scattering transform of the potential function $h$, first considered by Beals and Coifman \cite{beals-coifman} and Ablowitz and Nachman \cite{nachman-ablowitz} for the linear Schr\"odinger operator
$$
T_h(\xi) := \int_{\R^2}{\rme^{\rmi(x,\xi)}h(x,e_0|(1+R(x,z)|)(1+R(x,z))\rmd x},
\ |\xi|\ge \sqrt 2C_0, 
$$
$$
\ T_h(\xi) := 0,  \ |\xi|<\sqrt 2C_0,
$$ 
where constant $C_0$ from \eqref{constant for mod_z}, $e_0:=|\rme^{\rmi(x,z)}|=\rme^{\frac{1}{2}(x_1\xi_2-x_2\xi_1)}$, $z=\frac{1}{2}(\xi-\rmi J\xi)$ and 
$
J=\left(\begin{array}{rr}
0 & 1 \\
-1 & 0
\end{array}\right)
$. 
We denote this new constant $\sqrt 2 C_0$ (for simplicity) again by $C_0$.
 
It can be obtained that 
$$
T_h(\xi) = \lim_{j\to\infty}T_{h,j}(\xi)=\lim_{j\to\infty}\int_{\R^2}{\rme^{\rmi(x,\xi)}h(x,e_0|(1+R_j(x,z))|)(1+R_j(x,z))\rmd x},
$$
where the limit is uniform in $\xi \in \R^2$.

By following the same technique as in Gilbarg and Trudinger \cite{GT} (see theorems 11.4 and 11.8, pp. 281-288) we may prove that the Dirichlet boundary value problem for equation (3) with the boundary condition $f$ from Sobolev space $W^t_2(\partial\Omega), t>\frac{1}{2},$ has a unique solution $u$ from Sobolev space $W^s_2(\Omega), s>1.$ Thus, we may define the Dirichlet-to-Neumann map by $\Lambda_hf=\partial_{\nu}u$, where $\nu$ denotes the outward normal vector at the boundary and $\partial_{\nu}$ is the normal derivative. The main fact here is (we use the classical scheme as it is in linear case): the scattering amplitude $A(k_0,\vartheta',\vartheta)$ with fixed $k_0^2>0$ uniquely determines the Dirichlet-to-Neumann map $\Lambda _{h-k_0^2}$ (see\cite{novikov}, \cite{Sy1} and \cite{Sy2}  for details) and the Dirichlet-to-Neumann map in turn uniquely determines the scattering transform $T_h$ as a function of $\xi$. This fact is proved in the following lemma.

\begin{lemma}\label{lemma1}
The scattering transform $T_h(\xi)$ depends only on boundary values of $R$ and $\frac{\partial R}{\partial \nu}$ i.e. the Dirichlet-to-Neumann map determines the scattering transform uniquely.
\end{lemma}
\begin{proof}
Recalling \eqref{Requation} and denoting by $r=|\rme^{\rmi(\cdot,z)}(1+R)|$ we have
\begin{align*}
\Delta R+2\rmi(z,\nabla)R & =h(x,r)(R+1)\Rightarrow\\
\int_{\Omega}{\rme^{\rmi(x,\xi)}\Delta R + \rme^{\rmi(x,\xi)}2\rmi(z,\nabla)R\rmd x} & = \int_{\Omega}{\rme^{\rmi(x,\xi)}h(x,r)(R+1)\rmd x} = : T_h(\xi).
\end{align*}
Making use of the formulas 
$$
\nabla(2\rmi z\rme^{\rmi(x,\xi)} R)=R\Delta\rme^{\rmi(x,\xi)} + 2\rmi z\rme^{\rmi(x,\xi)} \nabla R
,\quad 2(z,\xi)= (\xi,\xi)
$$
and the Green's identity we derive
$$
T_h(\xi)=\int_{\partial\Omega}{\Big(\rme^{\rmi(x,\xi)}\frac{\partial R}{\partial \nu} - R \frac{\partial\rme^{\rmi(x,\xi)}}{\partial \nu} + 2\rmi (z,\nu)\rme^{\rmi(x,\xi)} R\Big)\rmd\sigma(x)}.
$$
\end{proof}

We are now in position to define the Born approximation of the potential function $h$.
\begin{defn}\label{def1}
The inverse fixed energy scattering Born approximation $q_B^f(x)$ of the potential function $h$ is defined by
$$
q_B^f(x):= \F^{-1}(T_h(\xi))(x),
$$
where $\F^{-1}$ is the inverse Fourier transform and where the equality is understood in the sense of tempered distributions.
\end{defn}

We write the inverse  Born approximation as
$$
q_B^f(x)-h_0(x)=q_B^f(x)-q_{B,1}^f(x)+q_{B,1}^f(x)-h_0(x),
$$
where $h_0(x)$ is defined as
$$
h_0(x):= \F^{-1}\Big(T_{h,0}(\xi)\Big)(x)=\F^{-1}\Big(\int_{\R^2}{\rme^{\rmi(x,\xi)}h(x,e_0)\rmd x}\Big)(x).
$$

The main goal of this work to prove the following result.
\begin{thm}\label{thm2}(Main theorem) Under the foregoing and subsequent conditions for the potential function $h$ 
$$
q_B^f(x)- h_0(x) \in H^t(\R^2),
$$
for any $t<1$ and modulo $C^{\infty}(\R^2)$-functions. 
\end{thm}
\begin{rem}
The embedding theorem for Sobolev spaces says that the difference $q_B^f(x)- h_0(x)$ belongs to $L^q(\R^2)$ for any $q<\infty$ by modulo $C^{\infty}(\R^2)$-functions. It means that all singularities from $L^p_{loc}(\R^2), p<\infty,$ of unknown function $h_0$ can be obtained exactly by the Born approximation which corresponds to the inverse scattering problem with fixed positive energy.
\end{rem}

In what follows it will be required an explicit form of the term $R_1(x,z)$. In order to obtain such a form we use the so-called $\bar{\partial}$ approach to inverse scattering. If we choose $z$ as in Definition \ref{def1} and Lemma \ref{lemma1} then a straightforward computation shows that 
$$
2\bar{\partial}(2\partial + (\xi_2+\rmi\xi_1))R_1=h(x,e_0)
$$
where
$$
\bar{\partial}=\frac{1}{2}\Big(\partial_{x_1} + \rmi\partial_{x_2}\Big),\
\partial=\frac{1}{2}\Big(\partial_{x_1} - \rmi\partial_{x_2}\Big).
$$
This formula leads to 
\begin{equation}\label{R1explicit}
R_1=\frac{1}{2}(\xi_2+\rmi\xi_1)^{-1}(\bar{\partial}^{-1}h(x,e_0) - \rme^{-\rmi(x,\xi)}\partial^{-1}(\rme^{\rmi(x,\xi)}\partial\bar{\partial}^{-1}h(x,e_0))).
\end{equation}

The following results by Nirenberg and Walker \cite{nirenberg-walker}  will be used.
\begin{prop}\label{prop1}
Let $p>1$, $\delta \in \R$ with $-\frac{2}{p}<\delta<1-\frac{2}{p}.$ Then 
$$
(\bar\partial^{-1})f(\xi) =- \frac{1}{\pi}\int_{\R^2}{\frac{f(\eta)}{\eta-\xi}\rmd \eta}
$$ 
defines a bounded operator from $L^p_{\delta+1}(\R^2)$ to $W^1_{p,\delta+1}(\R^2).$
\end{prop}

\begin{prop}\label{prop2}
Let $p$, $\delta$ as above.Then 
$$
(\partial\bar\partial^{-1})f(\xi)= -\frac{1}{\pi}\text{p.v.}\int_{\R^2}{\frac{f(\eta)}{(\eta-\xi)^2}\rmd \eta}
$$ 
defines a bounded operator from $L^p_{\delta+1}(\R^2)$ to $L^p_ {\delta+1}(\R^2).$
\end{prop}

\begin{lemma}\label{lemma2}
Under the conditions for the potential function $h$ as in Theorem \ref{thm1}, the difference $q_B^f(x)-q_{B,1}^f(x)$ belongs to the Sobolev space $H^t(\R^2)$ for any $t<1$.
\end{lemma}
\begin{proof}      
\begin{align*}
\|q_B^f(x)-q_{B,1}^f(x)\|_{H^t(\R^2)}^2 & = \|(1+|\xi|^2)^{\frac{t}{2}}\F\big(q_B^f(x)-q_{B,1}^f(x)\big)(\xi)\|_{L^2(\R^2)}^2\\ &=
\int_{|\xi|>C_0}{(1+|\xi|^2)^t}\\ 
& \quad\quad 
{\Big|\int_{\R^2}{e^{\rmi(x,\xi)}(h(x,e_0|1+R|)(1+R) - h(x,e_0|1+R_1|)(1+R_1))}\rmd x\Big|^2\rmd\xi}. 
\end{align*}
Denoting the inner integral by $I$ we estimate
\begin{align*}
|I|\leq & \int_{\R^2}{\Big(|h(x,e_0|1+R|) - h(x,e_0|1+R_1|)| + |h(x,e_0|1+R|) - h(x,e_0|1+R_1|)||R_1|}\\ &\quad
+ |h(x,e_0|1+R|)||R-R_1|\Big)\rmd x\\ \leq &
\|\beta(x)\|_{L^1(\Omega)}\|R-R_1\|_{{L^\infty}(\R^2)} + \|\beta(x)\|_{L^1(\Omega)}\|R_1\|_{{L^\infty}(\R^2)}\|R-R_1\|_{{L^\infty}(\R^2)}\\ &\quad
 +\|\alpha(x)\|_{L^1(\Omega)}\|R-R_1\|_{{L^\infty}(\R^2)} \\ \leq &
\frac{c}{|\xi|^{2\gamma}},
\end{align*}
where $|\xi|>C_0$. Here we used the fact that a compact supported function from $L^2$-space belongs also to $L^1$-space. 
So, we have
$$
\|q_B^f(x)-q_{B,1}^f(x)\|^2_{H^t(\R^2)} \leq  c\int_{|\xi|>C_0}{\frac{(1+|\xi|^2)^t}{|\xi|^{4\gamma}}}\,d\xi.
$$
The last integral converges for $4\gamma-2t>2$ or $t<2\gamma-1$ and since $\gamma<1$ we  finally get that
$q_B^f(x)-q_{B,1}^f(x) \in H^t(\R^2)$ for $t<1$.
\end{proof}


In order to investigate the term $q_{B,1}^f - h_0$ from the Born approximation we assume a little bit more about the potential function $h$. Namely, we assume formerly the following Taylor expansion
$$
h(x,e_0(1+s)) =h(x,e_0)+ \partial _sh(x,e_0(1+s))|_{s=0} s+ \ohbig(\beta_1(x) s^2),
$$
where
$|\partial _sh(x,e_0(1+s))|_{s=0}|\le \beta_1(x)$ and $\ohbig(\beta_1(x) s^2)$ with $\beta_1(x)\in L^2(\Omega)$ and with small $s$ in the neighborhood of zero and where $\ohbig$ is uniform in $x\in \Omega$ and such $s$. 
\begin{lemma}\label{lemma3}
Under the above conditions for the potential function $h$ the term $q_{B,1}^f - h_0$ belongs to  $H^t(\R^2)$ for any $t<1$ modulo $C^{\infty}(\R^2)$.
\end{lemma}

\begin{proof}
Let us denote by $s=|1+R_1(x,z)|-1$. Since $R_1(x,z)\to 0$ as $|z|\to \infty$ uniformly in $x\in\R^2$ then $s$ is small enough. Recalling that $T_{h,1}(\xi)=0$ for $|\xi|<C_0$ and denoting by $\chi(\xi)$ the characteristic function of the set $\{\xi\in \R^2:|\xi|>C_0\}$ we obtain
\begin{align*}
q_{B,1}^f(x)&=\F^{-1}\Big(\chi(\xi) \int_{\R^2}{\rme^{\rmi(x,\xi)}h(x,e_0|1+R_1|)(1+R_1)\rmd x} \Big)(x)\\&=
 h_0(x)+\F^{-1}\Big((\chi(\xi)-1)T_{h,0}(\xi)\Big)(x)\\&\quad+ 
 \F^{-1}\Big(\chi(\xi) \int_{\R^2}{\rme^{\rmi(x,\xi)}\big(h(x,e_0(1+s)) - h(x,e_0) - \partial_sh(x,e_0(1+s))|_{s=0}s\big)\rmd x} \Big)(x)\\&\quad+
 \F^{-1}\Big(\chi(\xi)\int_{\R^2}{\rme^{\rmi(x,\xi)}\big(h(x,e_0(1+s))R_1  + \partial_sh(x,e_0(1+s))|_{s=0}s\big)\rmd x} \Big)(x)\\&=h_0(x)+
 \F^{-1}\Big((\chi(\xi)-1)T_{h,0}(\xi)\Big)(x) + \F^{-1}\Big(\chi(\xi)\int_{\R^2}{\rme^{\rmi(x,\xi)}\ohbig(\beta_1(x)s^2)\rmd x} \Big)(x) \\&\quad +
 \F^{-1}\Big( \chi(\xi)\int_{\R^2}{\rme^{\rmi(x,\xi)}\big(h(x,e_0(1+s))R_1  + \partial_sh(x,e_0(1+s))|_{s=0}s\big)\rmd x} \Big)(x)\\ &=:
h_0(x) + I_1+I_2+I_3.
\end{align*}
The term $I_1$ is the Fourier transform of compactly supported distribution and thus, it is $C^{\infty}$-function in $\R^2$.
Since $s^2=\ohbig(|R_1|^2)$ then 
$$
|\F{I_2}| \leq \frac{c}{|\xi|^{2\gamma}},\quad |\xi|>C_0,
$$
and $I_2$ can be estimated in the Sobolev space $H^t(\R^2 )$as follows
\begin{align*}
\norm{I_2}_{H^t(\R^2)}^2 = \|(1+|\xi|^2)^{\frac{t}{2}}\F I_2\|_{L^2(\R^2)}^2\leq c\int_{|\xi|>C_0}{\frac{(1+|\xi|^2)^t}{|\xi|^{4\gamma}} \rmd \xi}<\infty
\end{align*}
for any $t<1$.

For the term $I_3$ we have the representation
\begin{align*}
I_3&=\F^{-1}\Big(\chi(\xi) \int_{\R^2}{\rme^{\rmi(x,\xi)}\big(h(x,e_0|1+R_1|) R_1 + \partial_sh(x,e_0(1+s))|_{s=0}ReR_1+\ohbig(\beta_1(x)|R_1|^2)\big)\rmd x} \Big)(x)\\&=:
I_3^{'} + I_3^{''}+I_3^{'''},
\end{align*}
where we have used that
$$
s=|1+R_1|-1=\frac{2\text{Re} R_1+|R_1|^2}{|1+R_1|+1}=ReR_1+\ohbig(|R_1|^2).
$$
The term $I_3^{'''}$ belongs to the Sobolev space $H^t(\R^2)$ for any $t<1$. The proof is completely the same as for the term $I_2$.
Substituting the representation \eqref{R1explicit} for $R_1$ into $I_3^{'}$ we get

\begin{align*}
I_3^{'} & =  \F^{-1}\Big( \chi(\xi)\int_{\R^2}{\rme^{\rmi(x,\xi)}h(x,e_0|1+R_1|) R_1}\rmd x\Big)(x) \\ &=
\F^{-1}\Big(\chi(\xi) \int_{\R^2}{\rme^{\rmi(x,\xi)}h(x,e_0|1+R_1|) \frac{1}{2}(\xi_2+\rmi\xi_1)^{-1}\bar{\partial}^{-1}h(x,e_0)\Big)\rmd x}\\& 
-\F^{-1}\Big(\chi(\xi) \int_{\R^2}\frac{1}{2}(\xi_2+\rmi\xi_1)^{-1}h(x,e_0|1+R_1|)\partial^{-1}(\rme^{\rmi(x,\xi)}\partial\bar{\partial}^{-1}h(x,e_0))\rmd x\Big)(x)\\&=
\frac{c_1}{\bar{z}}\ast \F^{-1}\Big(\chi(\xi) \int_{\R^2}{\rme^{\rmi(x,\xi)}h(x,e_0|1+R_1|) \bar{\partial}^{-1}h(x,e_0) \rmd x}\Big)(x) \\ &\quad +
\frac{c_2}{\bar{z}}\ast \F^{-1}\Big(\chi(\xi) \int_{\R^2}{h(x,e_0|1+R_1|)\partial^{-1}\big(\rme^{\rmi(x,\xi)}\partial\bar{\partial}^{-1}h(x,e_0)\big) \rmd x}\Big)(x) \\ & 
=: w_1 + w_2,
\end{align*}
where $z=y_1+iy_2$, $\F(\frac{1}{z})(\xi)=\frac{1}{i\xi_1-\xi_2}$ and $c_1,c_2$ are some known constants.\\

We assume now in addition a little bit more about the nonlinearity $h$ (mostly in order to avoid the technical problems). Namely, we assume that for $|\xi|>C_0$ the following asymptotical representations hold:
$$
h(x,e_0)=\sum_{j=0}^\infty\frac{\alpha_j(x)}{|\xi|^j} \quad \text{and} \quad
h(x,e_0|1+R_1|)=\sum_{j=0}^\infty\frac{\alpha_j(x)}{|\xi|^j},
$$
where $|\alpha_j(x)|\leq \alpha(x)/2$ and $|\tilde{\alpha}_j(x)|\leq \alpha(x)/2,$ where $\alpha$ is the same as in Theorem \ref{thm1}.\\
These assumptions can also be written as follows
$$
h(x,e_0) = \alpha_0(x) + \frac{\alpha(x)}{|\xi|}\ohbig(1)\quad \text{and} \\
$$
$$
h(x,e_0|1+R_1|) = \tilde{\alpha}_0(x) + \frac{\alpha(x)}{|\xi|}\ohbig(1),
$$
where
$\ohbig(1)$ depends uniformly on $x$.

Substituting the above formulas into $w_1$ yields
\begin{align*}
w_1 & = \frac{c}{\bar{z}}\ast \F^{-1}\Big(\chi(\xi)\int_{\R^2}{\rme^{\rmi(x,\xi)}\Big(\tilde{\alpha}_0(x) + \frac{\alpha(x)}{|\xi|}\ohbig(1) \Big)\bar{\partial}^{-1}\Big( \alpha_0(x) + \sum\limits_{j=1}^{\infty}\frac{\alpha_j(x)}{|\xi|^j}\Big) \rmd x}\Big)(x)\\&=
\frac{c}{\bar{z}}\ast \tilde{\alpha}_0(x)\bar{\partial}^{-1}\alpha_0(x)+\frac{c}{\bar{z}}\ast \F^{-1}(\chi-1)\ast \tilde{\alpha}_0(x)\bar{\partial}^{-1}\alpha_0(x)\\&\quad +
\frac{c}{\bar{z}}\ast \F^{-1}\Big(\frac{\chi(\xi)}{|\xi|}\int_{\R^2}{\rme^{\rmi(x,\xi)}\tilde{\alpha}_0(x) \bar{\partial}^{-1}\sum\limits_{j=1}^{\infty}\frac{\alpha_j(x)}{|\xi|^{j-1}} \rmd x}\Big)(x)
\\&\quad+
\frac{c}{\bar{z}}\ast \F^{-1}\Big(\frac{\chi(\xi)}{|\xi|}\int_{\R^2}{\rme^{\rmi(x,\xi)}\alpha(x) \ohbig(1)\bar{\partial}^{-1}\alpha_0(x) \rmd x}\Big)(x)
\\&\quad+
\frac{c}{\bar{z}}\ast \F^{-1}\Big(\frac{\chi(\xi)}{|\xi|^2}\int_{\R^2}{\rme^{\rmi(x,\xi)}\alpha(x) \ohbig(1)\bar{\partial}^{-1}\sum\limits_{j=1}^{\infty}\frac{\alpha_j(x)}{|\xi|^{j-1}}\rmd x}\Big)(x).
\end{align*}

Proposition \ref{prop1} and the Sobolev embedding theorem provide that the product\\ $\tilde{\alpha}_l(x)\bar{\partial}^{-1}\alpha_m(x),\ l,m=0,1,\ldots$ belongs to 
$$
L_{\comp}^2(\R^2)\cdot W^1_{2,\delta+1}(\R^2)\subset L^t_{\comp}(\R^2)
$$ 
for any $1\le t<2$. Thus, by Proposition \ref{prop1} the first term belongs to $W^1_{t,\delta+1}(\R^2),\ t<2.$ Here $1+\delta$ can be chosen actually equal to $0$ since $t<2$. For the second term in this sum we first remark that $\F^{-1}(\chi-1)$ belongs to the weighted space $L^2_{\sigma}(\R^2)$ with $\sigma<\frac{1}{2}$ and therefore it belongs to usual Lebesgue space $L^p(\R^2)$ with $\frac{4}{3}<p<2$. This fact implies (due to Hausdorff-Young inequality for convolution) that 
$$
\F^{-1}(\chi-1)\ast \tilde{\alpha}_0(x)\bar{\partial}^{-1}\alpha_0(x)\in L^r(\R^2)
$$
with any $\frac{4}{3}<r<2$. Hence, application of Proposition \ref{prop1} gives again that the second term belongs to $W^1_{t}(\R^2),\ \frac{4}{3}< t<2$.
The rest terms belong to $W^1_{2,\sigma}(\R^2)$ with any $\sigma>0$. It follows from the fact that the integrand in the inverse Fourier transform belongs to $L^2(\R^2)$. Then the  application of the Hausdorff-Young inequality for Fourier transform and Proposition \ref{prop1} gives us that $w_1 \in W^1_{r}(\R^2),r<2.$

We now turn to the integral in $w_2$ which can be investigated somehow by the same manner as $w_1$. Indeed, let us first write the integral in $x$ as follows
\begin{align*}
\int_{\R^2}{\big(\tilde{\alpha}_0(x) + \frac{\alpha(x)}{|\xi|}\ohbig(1)\big)\partial^{-1}\big(\rme^{\rmi(x,\xi)}\partial\bar{\partial}^{-1}\big(\alpha_0(x) + \sum\limits_{j=1}^{\infty}\frac{\alpha_j(x)}{|\xi|^j}\big)\big) \rmd x}\\=
\int_{\R^2}{\tilde{\alpha}_0(x)\partial^{-1}\big(\rme^{\rmi(x,\xi)}\partial\bar{\partial}^{-1}\alpha_0(x)\big) \rmd x} \\+
\frac{1}{|\xi|}\int_{\R^2}{\tilde{\alpha}_0(x)\partial^{-1}\big(\rme^{\rmi(x,\xi)}\partial\bar{\partial}^{-1}\sum\limits_{j=1}^{\infty}\frac{\alpha_j(x)}{|\xi|^{j-1}}\big) \rmd x}\\+
\frac{1}{|\xi|} \int_{\R^2}{\alpha(x)\ohbig(1)\partial^{-1}\big(\rme^{\rmi(x,\xi)}\partial\bar{\partial}^{-1}\alpha_0(x)\big) \rmd x}\\ +
\frac{1}{|\xi|^2} \int_{\R^2}{\ohbig(1)\alpha(x)\partial^{-1}\big(\rme^{\rmi(x,\xi)}\partial\bar{\partial}^{-1}\sum\limits_{j=1}^{\infty}\frac{\alpha_j(x)}{|\xi|^{j-1}}\big) \rmd x}.
\end{align*}

Each integral above can be treated (as the first one) as follows. According to Proposition \ref{prop1} we have
\begin{align*}
\int_{\R^2}{\tilde{\alpha}_0(x)\partial^{-1}\big(\rme^{\rmi(x,\xi)}\partial\bar{\partial}^{-1}\alpha_0(x)\big) \rmd x}\\&= 
-\frac{1}{\pi}\int_{\R^2}{\tilde{\alpha}_0(x)\int_{\R^2}{\frac{\rme^{\rmi(\eta,\xi)}\partial\bar{\partial}^{-1}\alpha_0(\eta)}{\bar{\eta}-\bar{x}} \rmd\eta}\rmd x}\\&=
-\frac{1}{\pi}\int_{\R^2}{\rme^{\rmi(\eta,\xi)}\partial\bar{\partial}^{-1}\alpha_0(\eta)\int_{\R^2}{\frac{\tilde{\alpha}_0(x)}{\bar{\eta}-\bar{x}} \rmd x}\rmd\eta}\\&=
-\int_{\R^2}{\rme^{\rmi(\eta,\xi)}\partial\bar{\partial}^{-1}\alpha_0(\eta)\partial^{-1}\tilde{\alpha}_0(\eta)\rmd\eta}\\&=
-\F\Big(\partial\bar{\partial}^{-1}\alpha_0(x)\partial^{-1}\tilde{\alpha}_0(x)\Big)(\xi).
\end{align*}
We used the Fubini's theorem since
$$
\int_{\R^2}{|\tilde{\alpha}_0(x)|\int_{\R^2}{\Big|\frac{\partial\bar{\partial}^{-1}\alpha_0(\eta)}{\bar{\eta}-\bar{x}}\Big| \rmd\eta}\rmd x}<\infty.
$$
Hence,

\begin{align*}
w_2 &=\frac{c}{\bar{z}}\ast \partial\bar{\partial}^{-1}\alpha_0(x)\partial^{-1}\tilde{\alpha}_0(x)+\frac{c}{\bar{z}}\ast \F^{-1}(\chi-1)\ast \partial\bar{\partial}^{-1}\alpha_0(x)\partial^{-1}\tilde{\alpha}_0(x)\\&\quad+
\frac{c}{\bar{z}}\ast\F^{-1}\Big(\frac{\chi(\xi)}{|\xi|} \int_{\R^2}{\rme^{\rmi(x,\xi)}\partial\bar{\partial}^{-1}\sum\limits_{j=1}^{\infty}\frac{\alpha_j(x)}{|\xi|^{j-1}}\partial^{-1}\tilde{\alpha}_0(x)\rmd x}\Big)(x)\\&\quad+
\frac{c}{\bar{z}}\ast\F^{-1}\Big(\frac{\chi(\xi)}{|\xi|} \int_{\R^2}{\rme^{\rmi(x,\xi)}\partial\bar{\partial}^{-1}\alpha_0(x)\partial^{-1}\sum\limits_{j=1}^{\infty}\frac{\tilde{\alpha}_j(x)}{|\xi|^{j-1}}\rmd x}\Big)(x)\\&\quad+
\frac{c}{\bar{z}}\ast\F^{-1}\Big(\frac{\chi(\xi)}{|\xi|^2} \int_{\R^2}{\rme^{\rmi(x,\xi)}\partial\bar{\partial}^{-1}\sum\limits_{j=1}^{\infty}\frac{\alpha_j(x)}{|\xi|^{j-1}}\partial^{-1}\sum\limits_{j=1}^{\infty}\frac{\tilde{\alpha}_j(x)}{|\xi|^{j-1}}\rmd x}\Big)(x).
\end{align*}
Now applying Proposition \ref{prop1}, Proposition \ref{prop2} and Sobolev embedding theorem we have that the product $\partial\bar{\partial}^{-1}\alpha_l(x)\partial^{-1}\tilde{\alpha}_m(x),\ l,m=0,1,\ldots$ belongs to $L^q(\R^2)$ for any $q<2$. Further steps and the results are completely the same as for $w_1$. The term $I_3^{''}$ can be estimated by the same manner as the term $I_3^{'}$ if we assume the same additional condition for $\partial_sh(x,e_0(1+s))|_{s=0}$ as we have assumed for $h(x,e_0)$ (these two conditions are very connected to each other). Thus,
$$
q_{B,1}^f(x)-h_0(x) \in C^\infty(\R^2)+ H^t(\R^2)+W^1_r(\R^2)
$$
with $t<1$ and $r<2$. But $W^1_r(\R^2)$ for $r<2$ is embedded in $H^t(\R^2)$ for $t<1$, therefore
 Lemma \ref{lemma3} is completely proved.
\end{proof} 

Now we are in the position to prove the main theorem, i.e. Theorem \ref{thm2}.

\begin{proof} Since
$$
q_B^f(x)-h_0(x)=q_B^f(x)-q_{B,1}^f(x)+q_{B,1}^f(x)-h_0(x)
$$
then Lemma \ref{lemma2} and Lemma \ref{lemma3} imply that 
$$
q_B^f(x)-h_0(x)\in H^t(\R^2)
$$
with any $t<1$ by modulo $C^{\infty}(\R^2)$-functions.  Thus, Theorem \ref{thm2} is proved.
\end{proof}

\bigskip

\section*{Acknowledgments}

This work was supported by the Academy of Finland (Application No. 250215, Finnish Programme for Centres of Excellence in Research 2012-2017).

\end{document}